\documentclass{birkjour}
\usepackage[msc-links,abbrev,nobysame,nobysame]
{amsrefs}

\BibSpecAlias{incollection}{inproceedings}
\BibSpec{article}{%
    +{}  {\PrintAuthors}                {author}
    +{.} { }                            {title}
    +{.} { }                            {part}
    +{:} { }                            {subtitle}
    +{.} { \PrintContributions}         {contribution}
    +{.} { \PrintPartials}              {partial}
    +{.} { \emph}                       {journal}
    +{,} { \textbf}                            {volume}
    +{}  { \parenthesize}               {number}
    +{:} {}                             {pages}
    +{,} { \PrintDateB}                 {date}
    +{,} { }                            {status}
    +{.} { \PrintTranslation}           {translation}
    +{.} { Reprinted in \PrintReprint}  {reprint}
    +{.} { }                            {note}
    +{.} {}                             {transition}
}

\BibSpec{partial}{%
    +{}  {}                             {part}
    +{:} { }                            {subtitle}
    +{.} { \PrintContributions}         {contribution}
    +{.} { \emph}                       {journal}
    +{,} { \textbf}                            {volume}
    +{}  { \parenthesize}               {number}
    +{:} {}                             {pages}
    +{,} { \PrintDateB}                 {date}
}

\BibSpec{book}{%
    +{}  {\PrintPrimary}                {transition}
    +{.} { \emph}                       {title}
    +{.} { }                            {part}
    +{:} { \emph}                       {subtitle}
    +{.} { }                            {series}
    +{,} { \voltext}                    {volume}
    +{.} { Edited by \PrintNameList}    {editor}
    +{.} { Translated by \PrintNameList}{translator}
    +{.} { \PrintContributions}         {contribution}
    +{.} { }                            {publisher}
    +{.} { }                            {organization}
    +{,} { }                            {address}
    +{,} { \PrintEdition}               {edition}
    +{,} { \PrintDateB}                 {date}
    +{.} { }                            {note}
    +{.} {}                             {transition}
    +{.} { \PrintTranslation}           {translation}
    +{.} { Reprinted in \PrintReprint}  {reprint}
    +{.} {}                             {transition}
}

\BibSpec{collection.article}{%
    +{}  {\PrintAuthors}                {author}
    +{.} { }                            {title}
    +{.} { }                            {part}
    +{:} { }                            {subtitle}
    +{.} { \PrintContributions}         {contribution}
    +{.} { \PrintConference}            {conference}
    +{.} { \PrintBook}                  {book}
    +{.} { In \PrintEditorsA}              {editor}
    +{.} { \emph}                         {booktitle}
    +{,} { pages~}                      {pages}
    +{,} { }                            {publisher}
    +{,} { }                            {organization}
    +{,} { }                            {address}
    +{,} { \PrintDateB}                 {date}
    +{.} { \PrintTranslation}           {translation}
    +{.} { Reprinted in \PrintReprint}  {reprint}
    +{.} { }                            {note}
    +{.} {}                             {transition}
}

 \BibSpec{inproceedings}{%
     +{}  {\PrintAuthors}                {author}
     +{.} { }                            {title}
     +{.} { }                            {part}
     +{:} { }                            {subtitle}
     +{.} { \PrintContributions}         {contribution}
     +{.} { \PrintConference}            {conference}
     +{.} { \PrintBook}                  {book}
     +{.} { In \PrintEditorsA}              {editor}
     +{.} {  \emph}                         {booktitle}
     +{,} { pages~}                      {pages}
     +{,} { }                            {publisher}
     +{,} { }                            {organization}
     +{,} { }                            {address}
     +{,} { \PrintDateB}                 {date}
     +{.} { \PrintTranslation}           {translation}
     +{.} { Reprinted in \PrintReprint}  {reprint}
     +{.} { }                            {note}
     +{.} {}                             {transition}
 }

\BibSpec{conference}{%
    +{}  {}                        {title}
    +{}  {\PrintConferenceDetails} {transition}
}

\BibSpec{innerbook}{%
    +{.} { \emph}                       {title}
    +{.} { }                            {part}
    +{:} { \emph}                       {subtitle}
    +{.} { }                            {series}
    +{,} { \voltext}                    {volume}
    +{.} { Edited by \PrintNameList}    {editor}
    +{.} { Translated by \PrintNameList}{translator}
    +{.} { \PrintContributions}         {contribution}
    +{.} { }                            {publisher}
    +{.} { }                            {organization}
    +{,} { }                            {address}
    +{,} { \PrintEdition}               {edition}
    +{,} { \PrintDateB}                 {date}
    +{.} { }                            {note}
    +{.} {}                             {transition}
}

\BibSpec{report}{%
    +{}  {\PrintPrimary}                {transition}
    +{.} { \emph}                       {title}
    +{.} { }                            {part}
    +{:} { \emph}                       {subtitle}
    +{.} { \PrintContributions}         {contribution}
    +{.} { Technical Report }           {number}
    +{,} { }                            {series}
    +{.} { }                            {organization}
    +{,} { }                            {address}
    +{,} { \PrintDateB}                 {date}
    +{.} { \PrintTranslation}           {translation}
    +{.} { Reprinted in \PrintReprint}  {reprint}
    +{.} { }                            {note}
    +{.} {}                             {transition}
}

\BibSpec{thesis}{%
    +{}  {\PrintAuthors}                {author}
    +{,} { \emph}                       {title}
    +{:} { \emph}                       {subtitle}
    +{.} { \PrintThesisType}            {type}
    +{.} { }                            {organization}
    +{,} { }                            {address}
    +{,} { \PrintDateB}                 {date}
    +{.} { \PrintTranslation}           {translation}
    +{.} { Reprinted in \PrintReprint}  {reprint}
    +{.} { }                            {note}
    +{.} {}                             {transition}
}

\makeindex

\usepackage{amssymb}
\usepackage{hyperref}
\DeclareFontFamily{OT1}{cyr}{}
\DeclareFontShape{OT1}{cyr}{m}{n}
   {  <5> <6> <7> <8> <9> gen * wncyr
      <10> <10.95> <12> <14.4> <17.28> <20.74> <24.88> wncyr10}{}
\DeclareFontShape{OT1}{cyr}{m}{it}
    {
       <5> <6> <7> <8> <9> gen * wncyi
      <10> <10.95> <12> <14.4> <17.28> <20.74> <24.88>wncyi10
      }{}
\DeclareFontShape{OT1}{cyr}{m}{ss}
    {
       <5> <6> <7> <8> wncyss8
       <9> wncy9
      <10> <10.95> <12> <14.4> <17.28> <20.74> <24.88>wncyss10
      }{}
\DeclareFontShape{OT1}{cyr}{m}{sc}
    {
       <5> <6> <7> <8> <9> <10> <10.95> <12> <14.4> <17.28> <20.74> <24.88>wncysc10
      }{}
\DeclareFontShape{OT1}{cyr}{bx}{n}
   {
       <5> <6> <7> <8> <9> gen * wncyb
      <10> <10.95> <12> <14.4> <17.28> <20.74> <24.88>wncyb10
      }{}
\DeclareTextFontCommand{\textcyr}{\fontfamily{cyr}\selectfont}

\providecommand{\cprime}{'}
\usepackage[all]{xy}

\newtheorem{thm}{Theorem}%[section]

\newtheorem{cor}[thm]{Corollary}

\theoremstyle{definition}

\newtheorem{example}[thm]{Example}
\theoremstyle{remark}

\renewcommand{\theenumi}{\roman{enumi}}
\providecommand{\myeprint}[2]{E-print: \href{#1}{\texttt{#2}}}

\makeatletter
\providecommand{\diag}{\mathop {\operator@font diag}\nolimits}
\makeatother

\providecommand{\such}{\,\mid\,}
\providecommand{\map}[1]{\mathsf{#1}}
\providecommand{\norm}[2][\relax]{\left\|#2\right\|\ifx#1\relax\else_{#1}\fi}
\providecommand{\modulus}[2][\relax]{\left| #2 \right|\ifx#1\relax\else_{#1}\fi}

\providecommand{\uir}[3][0]{\ifcase #1{\rho^{#2}_{#3}}% 
\or {\breve{\rho}^{#2}_{#3}}%
\or {\tilde{\rho}^{#2}_{#3}}\fi}

\providecommand{\linv}[2][\relax]{\mathfrak{L}^{#2}_{#1}}

\providecommand{\spec}[1][]{\ensuremath{\mathbf{sp}}\,}

\providecommand{\SU}[1][1,1]{\ensuremath{\mathrm{SU}(#1)}}
\providecommand{\scalar}[3][\relax]{\left\langle #2,#3 
        \right\rangle\ifx#1\relax\else_{#1}\fi}
\providecommand{\Space}[3][]{\ensuremath{\mathbb{#2}^{#3}_{#1}{}}}
  \providecommand{\FSpace}[3][]{\ensuremath{\ifx#2l \ell_{#3}^{#1}{}\else
  \mathcal{#2}_{#3}^{#1}{}\fi}} 
\providecommand{\rmi}{\mathrm{i}}

\providecommand{\myh}{h}
\providecommand{\myhbar}{\hslash}

\providecommand{\algebra}[1]{\ensuremath{\mathfrak{#1}}}
\makeatletter
  \providecommand{\limlike}[1]{\mathop {\operator@font #1}}
  \providecommand{\loglike}[1]{\mathop {\operator@font #1}\nolimits}
\makeatother

\providecommand{\MR}[1]{\textbf{MR}~\href{http://www.ams.org/mathscinet-getitem?mr=#1}{\#~#1}}
\providecommand{\Zbl}[1]{\textbf{Zbl}~\href{http://www.emis.de:80/cgi-bin/zmen/ZMATH/en/zmathf.html?first=1&maxdocs=3&type=html&an=#1&format=complete}{\#~#1}}
\providecommand{\eprint}[2]{E-print: \href{#1}{\texttt{#2}}}
\providecommand{\modulus}[2][\relax]{\left| #2 \right|\ifx#1\relax\else_{#1}\fi}

\providecommand{\oper}[1]{\mathcal{#1}}
\def\ifundefined#1{\expandafter\ifx\csname#1\endcsname\relax}

\begin{document}

%-------------------------------------------------------------------------
% editorial commands: to be inserted by the editorial office
%
%\firstpage{1} \volume{228} \Copyrightyear{2004} \DOI{003-0001}
%\seriesextraline{This is the Concrete Title of this Book\br H.E. R and S.T.C. W, Eds.}
%\firstpage{1}

\title{Uncertainty and Analyticity}

%----------Author 1
\author[Vladimir V. Kisil]
{\href{http://www.maths.leeds.ac.uk/~kisilv/}{Vladimir V. Kisil}}

\address{%
School of Mathematics,
University of Leeds,
Leeds LS2\,9JT,
England\newline
Web: \url{http://www.maths.leeds.ac.uk/~kisilv/}}
\email{\href{mailto:kisilv@maths.leeds.ac.uk}{kisilv@maths.leeds.ac.uk}}
\thanks{On  leave from Odessa University.}
\subjclass{Primary 81P05; Secondary 22E27}

\keywords{Quantum mechanics, classical mechanics, Heisenberg
  commutation relations, observables, Heisenberg group,
  Fock--Segal--Bargmann space, \(\SU\), Hardy space}

\date{December 1, 2013}

\begin{abstract}
  We describe a connection between minimal uncertainty states and
  holomorphy-type conditions on the images of the respective wavelet
  transforms. The most familiar example is the Fock--Segal--Bargmann
  transform generated by the Gaussian, however, this also occurs under
  more general assumptions.
\end{abstract}

%%% ----------------------------------------------------------------------
\maketitle
%%% ----------------------------------------------------------------------
%\tableofcontents

\section{Introduction}
\label{sec:introduction}

There are two and a half main examples of reproducing kernel spaces of
analytic function. One is the Fock--Segal--Bargmann (FSB) space and
others (one and a half)---the Bergman and Hardy spaces. The first
space is generated by the Heisenberg
group~\citelist{\cite{Kisil11c}*{\S~7.3} \cite{Folland89}*{\S~1.6}},
two others---by the group \(\SU\)~\cite{Kisil11c}*{\S~4.2} (this
explains our way of counting).

Those spaces have the following properties, which make their study
particularly pleasant and fruitful:
\begin{enumerate}
\item \label{item:group}
  There is a group, which acts transitively on functions'
  domain.
\item \label{item:reproducing-kernel}
  There is a reproducing kernel.
\item \label{item:spac-cons-holom}
  The space consists of holomorphic functions.
\end{enumerate}
Furthermore, for FSB space there is the following property:
\begin{enumerate}
\item[iv.] \label{item:uncertainity}
  The reproducing kernel is generated by a function, which minimises
  the uncertainty for coordinate and momentum observables.
\end{enumerate}
It is known, that a transformation group is responsible for the
appearance of the reproducing
kernel~\cite{AliAntGaz00}*{Thm.~8.1.3}. This paper shows that the last
two properties are equivalent and connected to the group as well.

\section{The Uncertainty Relation}
\label{sec:uncert-relat}

In quantum mechanics~\cite{Folland89}*{\S~1.1}, an observable
(self-adjoint operator on a Hilbert space \(\FSpace{H}{}\)) \(A\) produces the
expectation value \(\bar{A}\) on a state (a unit vector)
\(\phi\in \FSpace{H}{}\) by \(\bar{A}=\scalar{A\phi}{\phi}\). Then,
the dispersion is evaluated as follow: 
\begin{equation}
  \label{eq:dispersion}
  \Delta_\phi^2(A)=\scalar{(A-\bar{A})^2\phi}{\phi}=
  \scalar{(A-\bar{A})\phi}{(A-\bar{A})\phi}=
  \norm{(A-\bar{A})\phi}^2.
\end{equation}
The next theorem links obstructions of exact simultaneous
measurements with non-commutativity of observables.
\begin{thm}[The Uncertainty relation]
  \label{th:uncertainty}
  If \(A\) and \(B\) are self-adjoint operators on a Hilbert space \(\FSpace{H}{}\), then
  \begin{equation}
    \label{eq:uncertainty-abstract}
    \textstyle
    \norm{(A-a)u}\norm{(B-b)u}\geq \frac{1}{2} \modulus{\scalar{(AB-BA)u}{u}},
  \end{equation}
  for any \(u\in \FSpace{H}{}\) from the domains of \(AB\) and \(BA\) and \(a\),
  \(b\in\Space{R}{}\). Equality holds precisely when \(u\) is a
  solution of \(((A-a)+\rmi r (B-b))u=0\) for some real \(r\).
\end{thm}
\begin{proof} The proof is well-known~\cite{Folland89}*{\S~1.3}, but
  it is short, instructive and relevant for the following discussion,
  thus we include it in full. We start from simple algebraic
  transformations:
  \begin{eqnarray}
    \scalar{(AB-BA)u}{u}&=&\scalar{(A-a)(B-b)-(B-b)(A-a))u}{u}\nonumber \\
    &=&\scalar{(B-b)u}{(A-a)u}-\scalar{(A-a))u}{(B-b)u}\nonumber \\
    &=&2\rmi \Im\scalar{(B-b)u}{(A-a)u}
  \end{eqnarray}
  Then by the Cauchy--Schwartz inequality:
  \begin{displaymath}
    \label{eq:uncert-rel-proof}
    \textstyle\frac{1}{2}\scalar{(AB-BA)u}{u}\leq
    \modulus{\scalar{(B-b)u}{(A-a)u}}\leq\norm{(B-b)u}\norm {(A-a)u}.
  \end{displaymath}
  The equality holds if and only if \((B-b)u\) and \((A-a)u\) are
  proportional by a \emph{purely imaginary} scalar.
\end{proof}
The famous application of the above theorem is the following
fundamental relation in quantum
mechanics. Recall~\cite{Folland89}*{\S~1.2}, that the one-dimensional
Heisenberg group \(\Space{H}{1}\) consists of points
\(%(s,z)=
(s,x,y)\in\Space{R}{3}\), %\(z=x+\rmi y\) 
with the group law:
  \begin{equation}
    \label{eq:H-n-group-law}
    \textstyle
    %(s,z)\cdot(s',z')=(s+s'+\frac{1}{2}\Im(z'\bar{z}),z+z') 
    (s,x,y)*(s',x',y')
    =(s+s'+\frac{1}{2}(xy'-x'y),x+x',y+y').
  \end{equation} 
This is a nilpotent step two Lie group. By the Stone--von Neumann
theorem~\cite{Folland89}*{\S~1.5}, any infinite-dimensional
unitary irreducible representation of  \(\Space{H}{1}\) is unitary
equivalent to the Schr\"odinger representation \(\uir{}{\myhbar}\) in
\(\FSpace{L}{2}(\Space{R}{})\) parametrised by the Planck constant
\(\myhbar\in\Space{R}{}\setminus\{0\}\). A physically consistent form
of \(\uir{}{\myhbar}\) is~\cite{Kisil10a}*{(3.5)}: 
\begin{equation}
  \label{eq:H1-schroedinger-rep-q-space}
  [\uir{}{\myhbar}(s,x,y) {f}\,](q)=e^{-2\pi\rmi\myhbar (s+xy/2)
    -2\pi\rmi x q}\,{f}(q+{\myhbar} y).  
%  \textstyle[\uir{}{\myhbar}(s,x,y)f](x')= \exp(2\pi \rmi (\myhbar
%  (-s+yx'-\frac{1}{2}xy)-qy))\, f(x'-x).
\end{equation}
Elements of the Lie algebra \(\algebra{h}_1\), corresponding to the
infinitesimal generators \(X\) and \(Y\) of one-parameters subgroups
\((0,t/(2\pi),0)\) and \((0,0,t)\) in \(\Space{H}{1}\), are
represented in~\eqref{eq:H1-schroedinger-rep-q-space} by the
(unbounded) operators \(M\) and \(D\) on
\(\FSpace{L}{2}(\Space{R}{})\):
\begin{equation}
  \label{eq:lie-algebra-repres}
  \textstyle
  M=-\rmi q,\quad D=\myhbar\frac{d}{dq},\quad \text{with the
    commutator} \quad [M,D]= \rmi\myhbar I.
\end{equation}
In the Schr\"odinger model of quantum mechanics,
\(f(q)\in\FSpace{L}{2}(\Space{R}{})\) is interpreted as a wave
function (a state) of a particle, with \(M\) and \(D\) are the
observables of its coordinate and momentum.
\begin{cor}[Heisenberg--Kennard uncertainty relation]
  For the coordinate \(M\) and momentum \(D\) observables we have the
  \emph{Heisenberg--Kennard uncertainty} relation:
  \begin{equation}
    \label{eq:heisenberg-uncertainty}
    % \Delta_\phi(M-a) \cdot \Delta_\phi(D-b) \geq \frac{\myhbar}{2}.
    \Delta_\phi(M) \cdot \Delta_\phi(D) \geq \frac{\myh}{2}.
  \end{equation}
  The equality holds if and only if
  % \(\phi(x)=e^{2\pi\rmi b x}e^{-\pi  r(x-a)^2}\), \(r\in\Space{R}{}\)
  \(\phi(q)=e^{-c q^2}\), \(c\in\Space[+]{R}{}\) is the vacuum state 
  in the Schr\"odinger model.
\end{cor}
\begin{proof}
  The relation follows from the commutator \([M,D]=\rmi\myhbar I\),
  which, in turn, is the representation of the Lie algebra
  \(\algebra{h}_1\) of the Heisenberg group. The minimal uncertainty
  state in the Schrodinger representation is a solution of the
  differential equation: \((M-\rmi r D)\phi=0\) for some
  \(r\in\Space{R}{}\), or, explicitly:
  \begin{equation}
    \label{eq:gaussian-uncertainty-eq}
    (M-\rmi r D)\phi=-\rmi\left(q+r{\myhbar}\frac{d}{dq}\right)\phi(q)=0.
  \end{equation}
  The solution is the Gaussian \(\phi(q)=e^{-c q^2}\),
  \(c=\frac{1}{2r\myhbar}\). For \(c>0\), this function  is in
  the state space \(\FSpace{L}{2}(\Space{R}{})\). 
\end{proof}
It is common to say that the Gaussian \(\phi(q)=e^{-c q^2}\) represents
the ground state, which minimises the uncertainty of coordinate and
momentum. 

\section{Wavelet transform and analyticity}
\label{sec:gaussian}

\subsection{Induced wavelet transform}
\label{sec:induc-wavel-transf}

The following object is common in quantum
mechanics~\cite{Kisil02e}, signal processing, harmonic
analysis~\cite{Kisil12d}, operator theory~\cites{Kisil12b,Kisil13a}
and many other areas~\cite{Kisil11c}.  Therefore, it has various
names~\cite{AliAntGaz00}: coherent states, wavelets, matrix
coefficients, etc.  In the most fundamental
situation~\cite{AliAntGaz00}*{Ch.~8}, we start from an irreducible
unitary representation \(\uir{}{}\) of a Lie group \(G\) in a Hilbert
space \(\FSpace{H}{}\). For a vector \(f\in \FSpace{H}{}\) (called mother wavelet, vacuum
state, etc.), we define the map \(\oper{W}_f\) from \(\FSpace{H}{}\) to a space
of functions on \(G\) by:
\begin{equation}
  \label{eq:wavelet-trans}
  [\oper{W}_f v](g)=\tilde{v}(g):=\scalar{v}{\uir{}{}(g)f}.
\end{equation}
Under the above assumptions, \(\tilde{v}(g)\) is a bounded continuous
function on \(G\). The map \(\oper{W}_f\) intertwines \(\uir{}{}(g)\)
with the left shifts on \(G\):
\begin{equation}
  \label{eq:left-shift-itertwine}
  \oper{W}_f \circ \uir{}{}(g) =\Lambda(g)\circ   \oper{W}_f ,\qquad
  \text{ where } \Lambda(g):\tilde{v}(g')\mapsto
  \tilde{v}(g^{-1}g').
\end{equation}
Thus, the image \(\oper{W}_f \FSpace{H}{}\) is invariant under the left shifts on
\(G\). If \(\uir{}{}\) is square integrable and \(f\) is
admissible~\cite{AliAntGaz00}*{\S~8.1}, then \(\tilde{v}(g)\) is
square-integrable with respect to the Haar measure on \(G\). At this
point, none of admissible vectors has an advantage over others.

It is common~\cite{Kisil11c}*{\S~5.1}, that there exists a closed subgroup
\(H\subset G\) and a respective \(f\in \FSpace{H}{}\) such that \(\uir{}{}(h)
f=\chi(h) f\) for some character \(\chi\) of \(H\). In this case, it
is enough to know values of \(\tilde{v}(\map{s}(x))\), for any
continuous section \(\map{s}\) from the homogeneous space \(X=G/H\) to
\(G\). The map \(v \mapsto \tilde{v}(x)=\tilde{v}(\map{s}(x))\)
intertwines \(\uir{}{}\) with the representation \(\uir{}{\chi}\) in a
certain function space on \(X\) induced by the character \(\chi\) of
\(H\)~\cite{Kirillov76}*{\S~13.2}. We call the map \(\oper{W}_f: v
\mapsto \tilde{v}(x)\) the \emph{induced wavelet
  transform}~\cite{Kisil11c}*{\S~5.1}.

For example, if \(G=\Space{H}{1}\), \(H=\{(s,0,0)\in\Space{H}{1}:\
s\in\Space{R}{}\}\) and its character
\(\chi_\myhbar(s,0,0)=e^{2\pi\rmi\myhbar s}\), then any vector
\(f\in\FSpace{L}{2}(\Space{R}{})\) satisfies \(\uir{}{\myhbar}(s,0,0)
f=\chi_\myhbar(s) f\) for the
representation~\eqref{eq:H1-schroedinger-rep-q-space}. Thus, we still
do not have a reason to prefer any admissible vector to
others. %However, the following result provides such a reason.

\subsection{Right shifts and analyticity}

To discover some preferable mother wavelets, we use the following a
general result from~\cite{Kisil11c}*{\S~5}.  Let \(G\) be a locally
compact group and \(\uir{}{}\) be its representation in a Hilbert
space \(\FSpace{H}{}\). Let \([\oper{W}_f v](g)=\scalar{v}{\uir{}{}(g)f}\) be the
wavelet transform defined by a vacuum state \(f\in \FSpace{H}{}\).  Then, the
right shift \(R(g): [\oper{W}_f v](g')\mapsto [\oper{W}_f v](g'g)\) for
\(g\in G\) coincides with the wavelet transform
\([\oper{W}_{f_g} v](g')=\scalar{v}{\uir{}{}(g')f_g}\) defined by the vacuum
state \(f_g=\uir{}{}(g) f\).  In other words, the covariant transform
intertwines%
\index{intertwining operator}%
\index{operator!intertwining} right shifts on the group \(G\) with the
associated action \(\uir{}{}\) on vacuum states,
cf.~\eqref{eq:left-shift-itertwine}:
\begin{equation}
  \label{eq:wave-intertwines-right}
  R(g) \circ \oper{W}_f= \oper{W}_{\uir{}{}(g)f}.
\end{equation}
Although, the above observation is almost trivial, applications of the following
corollary are not.
\begin{cor}[Analyticity of the wavelet transform,~\cite{Kisil11c}*{\S~5}] 
  \label{co:cauchy-riemann-integ}
  Let \(G\) be a group and \(dg\) be a measure on \(G\). Let
  \(\uir{}{}\) be a unitary representation of \(G\), which can be
  extended by integration to a vector space \(V\) of functions or
  distributions on \(G\).  Let a mother wavelet \(f\in \FSpace{H}{}\) satisfy the
  equation
  \begin{displaymath}
    \int_{G} a(g)\, \uir{}{}(g) f\,dg=0,
  \end{displaymath}
  for a fixed distribution \(a(g) \in V\). Then  any wavelet transform
  \(\tilde{v}(g)=\scalar{v}{\uir{}{}(g)f}\) obeys the condition:
  \begin{equation}
     \label{eq:dirac-op}
     D\tilde{v}=0,\qquad \text{where} \quad D=\int_{G} \bar{a}(g)\, R(g) \,dg,
  \end{equation}
  with \(R\) being the right regular representation of \(G\).
\end{cor}
% \begin{cor}[Analyticity of the wavelet transform,~\cite{Kisil11c}*{\S~5}] 
%   \label{co:cauchy-riemann}
%   Let \(\uir{}{*}\) be the derived representation on \(\FSpace{H}{}\) of a Lie algebra
%   \(\algebra{g}\) of \(G\).    Let a vacuum state be a null-solution, i.e.\ \(A f=0\),
%   for the operator \(A=\sum_j a_j \uir{X_j}{*}\), where
%   \(X_j\in\algebra{g}\) and \(a_j\) are constants.  Then the covariant
%   transform \(\tilde{ v}(g)=\scalar{v}{\uir{}{}(g)f}\) for any \(v\)
%   satisfies:
%   \begin{equation}
%     \label{eq:dirac-op}
%     D \tilde{v}(g)= 0, \qquad \text{where} \quad
%     D=\sum_j \bar{a}_j\linv{X_j}.
%   \end{equation}
%   and \(\linv{X_j}\) are the left invariant fields (Lie derivatives) on
%   \(G\) produced by \(X_j\).
% \end{cor}

Some applications (including discrete one) produced by the \(ax+b\) group
can be found in~\cite{Kisil12d}*{\S~6}.  We turn to the Heisenberg
group now.
\begin{example}[Gaussian and FSB transform]
  \label{ex:gaussian-fsb}
  The Gaussian \(\phi(x)=e^{- c q^2/2}\) is a null-solution of the
  operator \(\myhbar c M-\rmi D\). For the centre \(Z=\{(s,0,0):\
  s\in\Space{R}{}\}\subset \Space{H}{1}\), we define the section
  \(\map{s}:\Space{H}{1}/Z\rightarrow \Space{H}{1}\) by
  \(\map{s}(x,y)=(0,x,y)\). Then, the corresponding induced wavelet
  transform is:
  \begin{equation}
    \label{eq:fsb-transform}
   \tilde{v}(x,y)=\scalar{v}{\uir{}{}(\map{s}(x,y))f}= \int_{\Space{R}{}} v(q)\, e^{\pi\rmi\myhbar xy
    -2\pi\rmi x q}\,e^{-c(q+{\myhbar} y)^2/2}\,dq.
  \end{equation}
  The infinitesimal generators \(X\) and \(Y\) of one-parameters subgroups
  \((0,t/(2\pi),0)\) and \((0,0,t)\) are represented through the right
  shift in~\eqref{eq:H-n-group-law} by
  \begin{displaymath}
    \textstyle
    R_*(X)=-\frac{1}{4\pi}y\partial_s+\frac{1}{2\pi}\partial_x,\quad
    R_*(Y)=\frac{1}{2}x\partial_s+\partial_y.
  \end{displaymath}
  For the representation induced
  by the character \(\chi_\myhbar(s,0,0)=e^{2\pi\rmi\myhbar s}\) we
  have \(\partial_s= 2\pi\rmi\myhbar
  I\). Cor.~\ref{co:cauchy-riemann-integ} ensures that the operator
  \begin{equation}
    \label{eq:fsb-cauchy-riemann}
    \myhbar c\cdot R_*(X)+\rmi\cdot R_*(Y)=
   -\frac{\myhbar}{2}(2\pi x+ {\rmi\myhbar c}y) +\frac{\myhbar
      c}{2\pi}\partial_x +\rmi \partial_y
  \end{equation}
  annihilate any \(\tilde{v}(x,y)\) from~\eqref{eq:fsb-transform}.
  The integral~\eqref{eq:fsb-transform} is known as
  Fock--Segal--Bargmann (FSB) transform and in the most common case
  the values \(\myhbar=1\) and \(c=2\pi\) are used. For these,
  operator~\eqref{eq:fsb-cauchy-riemann} becomes \(-\pi(x+\rmi
  y)+(\partial_x+\rmi\partial _y)=-\pi z+2\partial_{\bar{z}}\) with
  \(z=x+\rmi y\). Then
  the function \(V(z)=e^{\pi z
    \bar{z}/2}\,\tilde{v}(z)=e^{\pi(x^2+y^2)/2}\, \tilde{v}(x,y)\)
  satisfies the Cauchy--Riemann equation \(\partial_{\bar{z}}
  V(z)=0\).
\end{example}
This example
shows, that the Gaussian is a preferred vacuum state (as producing
analytic functions through FSB transform) exactly for the same reason
as being the minimal uncertainty state: the both are derived from the
identity \((\myhbar c M+\rmi D)e^{- c q^2/2}=0\).

\subsection{Uncertainty and analyticity}
\label{sec:uncert-analyt}
The main result of this paper is a generalisation of the previous
observation, which bridges together Cor.~\ref{co:cauchy-riemann-integ} and
Thm.~\ref{th:uncertainty}. Let \(G\), \(H\), \(\uir{}{}\) and
\(\FSpace{H}{}\) be
as before. Assume, that the homogeneous space
\(X=G/H\) has a (quasi-)invariant measure
\(d\mu(x)\)~\cite{Kirillov76}*{\S~13.2}. Then, for a function (or a
suitable distribution) \(k\) on \(X\) we can define the integrated
representation:
\begin{equation}
  \label{eq:integrated-rep}
  \uir{}{}(k)=\int_X k(x)\uir{}{}(\map{s}(x))\,d\mu(x),
\end{equation}
which is (possibly, unbounded) operators on (possibly, dense subspace
of) \(\FSpace{H}{}\). In particular, \(R(k)\) denotes the
integrated right shifts,  for \(H=\{e\}\).

\begin{thm}
  Let \(k_1\) and \(k_2\) be two distributions on
  \(X\) with the respective integrated representations
  \(\uir{}{}(k_1)\) and \(\uir{}{}(k_2)\).  The following are
  equivalent:
  \begin{enumerate}
  \item A vector \(f\in \FSpace{H}{}\) satisfies the identity
    \begin{displaymath}
      \Delta_f(\uir{}{}(k_1))\cdot
    \Delta_f(\uir{}{}(k_2))=\modulus{\scalar{[\uir{}{}(k_1),\uir{}{}(k_1)]
      f}{f}}.
    \end{displaymath}
  \item The image of the wavelet transform \(\oper{W}_f: v \mapsto
    \tilde {v}(g)=\scalar{v}{\uir{}{}(g)f}\) consists of functions
    satisfying the equation \(R(k_1+\rmi r k_2) \tilde {v}=0\) for
    some \(r\in\Space{R}{}\), where \(R\) is the integrated
    form~\eqref{eq:integrated-rep} of the right regular representation
    on \(G\).
  \end{enumerate}
\end{thm}
\begin{proof}
  This is an immediate consequence of a combination of
  Thm.~\ref{th:uncertainty} and Cor.~\ref{co:cauchy-riemann-integ}.
\end{proof}
Example~\ref{ex:gaussian-fsb} is a particular case of this theorem
with \(k_1(x,y)=\delta'_x(x,y)\) and \(k_2(x,y)=\delta'_y(x,y)\)
(partial derivatives of the delta function), which represent vectors
\(X\) and \(Y\) from the Lie algebra \(\algebra{h}_1\). The next
example will be of this type as well.

\subsection{Hardy space}
\label{sec:hardy-space}

Let \(\SU\) be the group of \(2\times 2\) complex matrices of the
form \(\begin{pmatrix}
  \alpha&\beta \\\bar{\beta}&\bar{\alpha}
\end{pmatrix}\) with the unit determinant
\(\modulus{\alpha}^2-\modulus{\beta}^2=1\). A standard basis in the
Lie algebra \(\algebra{su}_{1,1}\) is
\begin{displaymath}
  A=\frac{1}{2}
  \begin{pmatrix}
    0&-\rmi\\\rmi&0
  \end{pmatrix},\quad
  B=\frac{1}{2}
  \begin{pmatrix}
    0&1\\1&0
  \end{pmatrix},\quad
  Z=
  \begin{pmatrix}
    \rmi&0\\0&-\rmi
  \end{pmatrix}.
\end{displaymath}
The respective one-dimensional subgroups consist of matrices:
\begin{displaymath}
  e^{tA}=
  \begin{pmatrix}
    \cosh \frac{t}{2} & -\rmi\sinh \frac{t}{2}\\ \rmi\sinh \frac{t}{2}& \cosh \frac{t}{2}
  \end{pmatrix},\ 
  e^{tB}=
  \begin{pmatrix}
    \cosh \frac{t}{2} & \sinh \frac{t}{2}\\ \sinh \frac{t}{2}& \cosh \frac{t}{2}
  \end{pmatrix},\ 
  e^{tZ}=
  \begin{pmatrix}
    e^{\rmi t}&0\\0&e^{-\rmi t}
  \end{pmatrix}.
\end{displaymath}
The last subgroup---the maximal compact subgroup of \(\SU\)---is
usually denoted by \(K\).  The commutators of the
\(\algebra{su}_{1,1}\) basis elements are
\begin{equation}
  \label{eq:su11-commutators}
  [Z,A]=2B, \qquad [Z,B]=-2A, \qquad [A,B]=- \frac{1}{2} Z.
\end{equation}

Let \(\Space{T}{}\) denote the unit circle in \(\Space{C}{}\) with
the rotation-invariant measure. The
mock discrete representation of \(\SU\)~\cite{Lang85}*{\S~VI.6} acts
on \(\FSpace{L}{2}(\Space{T}{})\) by unitary transformations
\begin{equation}
  \label{eq:SU-action}
  [\uir{}{1}(g) f](z) 
  % = \frac{\modulus{\bar{\beta}z+\alpha }^2}{(\bar{\beta}z+\alpha)^2 }\, 
  % f\left(\frac{\bar{\alpha}z+\beta }{\bar{\beta} z+\alpha    }\right)
  = \frac{1}{(\bar{\beta}z+\bar{\alpha}) }\, 
  f\left(\frac{{\alpha}z+\beta }{\bar{\beta} z+\bar{\alpha}    }\right),
  % = \frac{\beta\bar{z}+\bar{\alpha }}{\bar{\beta}z+\alpha}\, 
  % f\left(\frac{\bar{\alpha}z+\beta }{\bar{\beta} z+\alpha  }\right).
  \qquad
  g^{-1}=
  \begin{pmatrix}
    \alpha&\beta\\
    \bar{\beta}&\bar{\alpha}
  \end{pmatrix}.
\end{equation}
The respective derived representation \(\uir{}{1*}\) of the
\(\algebra{su}_{1,1}\) basis is:
\begin{displaymath}
\textstyle
  \uir{A}{1*}= \frac{\rmi}{2}(z+(z^2+1)\partial_z),\quad
  \uir{B}{1*}= \frac{1}{2}(z+(z^2-1)\partial_z),\quad
  \uir{Z}{1*}=-\rmi I-2\rmi z\partial_z.
\end{displaymath}
Thus, \(\uir{B+\rmi A}{1*}=-\partial_z\) and the function
\(f_+(z)\equiv 1\) satisfies \(\uir{B+\rmi A}{1*} f_+=0\). Recalling
the commutator \([A,B]=-\frac{1}{2}Z\) we note that
\(\uir{}{1}(e^{tZ})f_+=e^{\rmi t} f_+\). Therefore, there is the
following identity for dispersions on this state:
\begin{displaymath}
  \textstyle
  \Delta_{f_+}(\uir{A}{1*})\cdot \Delta_{f_+}(\uir{B}{1*}) =\frac{1}{2}, 
\end{displaymath}
with the minimal value of uncertainty among all eigenvectors of the operator
\(\uir{}{1}(e^{tZ})\).  

Furthermore, the vacuum state \(f_+\) generates the induced wavelet
transform for the subgroup \(K=\{e^{tZ} \such t\in\Space{R}{}\}\).  We
identify \(\SU/K\) with the open unit disk%
\index{unit!disk}%
\index{disk!unit} \(D=\{w\in\Space{C}{} \such \modulus{w}<1
\}\)~\citelist{\cite{Kisil11c}*{\S~5.5} \cite{Kisil13a}}.  The map
\(\map{s}: \SU/K \rightarrow \SU\) is defined as
\(\map{s}(w)=\frac{1}{\sqrt{1-\modulus{w}^2}}
\begin{pmatrix}
  1&w\\
  \bar{w}&1
\end{pmatrix}\). Then, the induced wavelet transform is:
\begin{align*}
  \tilde{v}(w)=\scalar{v}{\uir{}{1}(\map{s}(w)) f_+}&=
  \frac{1}{2\pi\sqrt{1-\modulus{w}^2}}\int_{\Space{T}{}}
  \frac{v(e^{\rmi \theta})\,d\theta}{1-w e^{-\rmi \theta}}\\
  &=
    \frac{1}{2\pi\rmi\sqrt{1-\modulus{w}^2}}\int_{\Space{T}{}}
    \frac{v(e^{\rmi \theta})\,de^{\rmi \theta}}{e^{\rmi \theta}-w }. 
\end{align*}
Clearly, this is the Cauchy integral up to the factor
\(\frac{1}{\sqrt{1-\modulus{w}^2}}\), which presents the conformal
metric on the unit disk.  Similarly, we can consider the operator
\(\uir{B-\rmi A}{1*}=z+z^2\partial_z\) and the function
\(f_-(z)=\frac{1}{z}\) simultaneously solving the equations
\(\uir{B-\rmi A}{1*} f_-=0\) and \(\uir{}{1}(e^{tZ})f_-=e^{-\rmi t}
f_-\). It produces the integral with the conjugated Cauchy kernel.

Finally, we can calculate the operator~\eqref{eq:dirac-op}
annihilating the image of the wavelet transform. In the coordinates
\((w,t)\in (\SU/K)\times K\), the restriction to the induced
subrepresentation is, cf.~\cite{Lang85}*{\S~IX.5}:
\begin{align*}
\linv{ B-\rmi  A}{}
= e^{2 \rmi  t} (-\frac{1}{2} w+(1- \modulus{w}^2) \partial_{\bar{w}}).
\end{align*}
Furthermore, if \(\linv{B- \rmi A}{}\tilde{v}(w)=0\), then
\(\partial_{\bar{w}}(\sqrt{1-w\bar{w}}\cdot \tilde{v}(w))=0\). That is,
\(V(w)=\sqrt{1-w\bar{w}}\cdot \tilde{v}(w)\) is a holomorphic
function on the unit disk.

Similarly, we can treat representations of \(\SU\) in the space of
square integrable functions on the unit disk. The irreducible
components of this representation are isometrically
isomorphic~\cite{Kisil11c}*{\S~4--5} to the weighted Bergman spaces of
(purely poly-)analytic functions on the unit, cf.~\cite{Vasilevski99}.

%\section*{Acknowledgment}

\small
%\bibliography{abbrevmr,akisil,analyse,algebra,arare,aclifford,aphysics,acompute,ageometry,acombin}
\providecommand{\noopsort}[1]{} \providecommand{\printfirst}[2]{#1}
  \providecommand{\singleletter}[1]{#1} \providecommand{\switchargs}[2]{#2#1}
  \providecommand{\irm}{\textup{I}} \providecommand{\iirm}{\textup{II}}
  \providecommand{\vrm}{\textup{V}} \providecommand{\cprime}{'}
  \providecommand{\eprint}[2]{\texttt{#2}}
  \providecommand{\myeprint}[2]{\texttt{#2}}
  \providecommand{\arXiv}[1]{\myeprint{http://arXiv.org/abs/#1}{arXiv:#1}}
  \providecommand{\doi}[1]{\href{http://dx.doi.org/#1}{doi:
  #1}}

% ------------------------------------------------------------------------
\end{document}
% ------------------------------------------------------------------------